\documentclass{article}

\usepackage{arxiv}

\usepackage[utf8]{inputenc} % allow utf-8 input
\usepackage[T1]{fontenc}    % use 8-bit T1 fonts
\usepackage{hyperref}       % hyperlinks
\usepackage{url}            % simple URL typesetting
\usepackage{booktabs}       % professional-quality tables
\usepackage{amsfonts}       % blackboard math symbols
\usepackage{nicefrac}       % compact symbols for 1/2, etc.
\usepackage{microtype}      % microtypography
\usepackage{lipsum}
\usepackage{graphicx}
\graphicspath{ {./images/} }
\usepackage{algorithm}%
\usepackage{algorithmicx}%
\usepackage{algpseudocode}%
\usepackage{listings}%
\usepackage{amsmath}
\usepackage{amsthm}  % <-- add this line

\newtheorem{lemma}{Lemma}
\newtheorem{definition}{Definition}%
\newtheorem{theorem}{Theorem}%  meant for 

\title{Graph-based method for constructing consensus trees}

\author{
 Elio Torquet \\
  Department of Bioinformatics \\
  University of Montpellier \\
  163 rue Auguste Broussonnet \\ 
  34000 Montpellier, France \\
  \texttt{elio.torquet@etu.umontpellier.fr} \\
  \And
 Jesper Jansson \\
  Department of Communications and Computer Engineering \\
  Kyoto University \\
  Yoshida-Honmachi, 606-8501 Tokyo, Japan \\
  \texttt{jj@i.kyoto-u.ac.jp} \\
  \And
 Nadia Tahiri \\
  Department of Computer Science \\
  University of Sherbrooke \\
  2500 Boulevard de l'Université, Sherbrooke, QC, Canada \\
  \texttt{Nadia.Tahiri@USherbrooke.ca} \\
}

\begin{document}

\maketitle

\begin{abstract}
A consensus tree is a phylogenetic tree that synthesizes a given collection of phylogenetic trees, all of which share the same leaf labels but may have different topologies, typically obtained through bootstrapping. Our research focuses on creating a consensus tree from a collection of phylogenetic trees, each detailed with branch-length data. We integrate branch lengths into the consensus to encapsulate the progression rate of genetic mutations. However, traditional consensus trees, such as the strict consensus tree, primarily focus on the topological structure of these trees, often neglecting the informative value of branch lengths. This oversight disregards a crucial aspect of evolutionary study and highlights a notable gap in traditional phylogenetic approaches. In this paper, we extend \textit{PrimConsTree}\footnote{A preliminary version of this article was presented at \emph{the Fifteenth International Conference on Bioscience, Biochemistry, and Bioinformatics (ICBBB~2025)}~(reference~\cite{torquet2005icbbb}).}, a graph-based method for constructing consensus trees. This algorithm incorporates topological information, edge frequency, clade frequency, and branch length to construct a more robust and comprehensive consensus tree. Our adaptation of the well-known Prim algorithm efficiently identifies the maximum frequency branch and maximum frequency nodes to build the optimal consensus tree. This strategy was pre-processed with clustering steps to calibrate the robustness and accuracy of the consensus tree.\\
\textbf{Availability and implementation:} The source code of PrimConsTree is freely available on GitHub at \url{https://github.com/tahiri-lab/PrimConsTree}.
\end{abstract}

%%================================%%
%% Sample for structured abstract %%
%%================================%%


\maketitle

\section{Introduction}
\label{sec:intro}
The exploration of phylogenetic relationships is crucial in biology to reveal evolutionary connections among species. A consensus tree is a computational method used to synthesize a set of phylogenetic trees, aiming to distill the most frequently occurring characteristics into a single representative tree. This approach facilitates the identification of common evolutionary relationships and patterns across multiple phylogenetic analyses \cite{o2018efficacy}. In this context, the accuracy of a consensus tree lies in its ability to effectively represent the collective input while accounting for uncertainty or divergence in the data. However, consolidating multiple trees into a single structure presents a significant challenge \cite{degnan2009properties}.

The difficulty in this process is reconciling variations and conflicts that may arise from a set of phylogenetic trees. Integrating diverse evolutionary perspectives and resolving inconsistencies becomes intricate due to structural differences between trees, often making them incompatible. This mismatch requires careful consideration and the application of advanced computational techniques to construct a coherent and accurate composite consensus tree. In addressing this issue, a proposed solution recommends employing multiple consensus trees for a more comprehensive approach \cite{tahiri2022building,aguse2019summarizing,tahiri2018new}.

The phylogenetic tree involves three main elements: 1) topology, 2) branch length, and 3) label position. In phylogenetic trees, topology refers to the branching structure that represents evolutionary relationships among species. It is important to note that closely related species based on topology may not necessarily exhibit morphological similarity. For instance, crocodiles are more closely related to birds than to lizards based on their evolutionary lineage, despite crocodiles and lizards appearing more morphologically similar. This occurs due to differing rates of evolutionary change, particularly rapid morphological evolution along the bird lineage. Branch lengths in phylogenetic trees represent the amount of genetic divergence, commonly measured in nucleotide substitutions. While branch lengths can sometimes be interpreted as indicative of time, this interpretation is contingent upon the assumption of a molecular clock. In cases where terminal branches descend from the same common ancestor, differences in branch lengths are generally interpreted as variations in evolutionary rates rather than direct representations of temporal duration. For instance, a longer branch may reflect a higher rate of substitution rather than a longer period of time when compared to a shorter branch arising from the same ancestral node.
The position of the labels, affixed to the ends of the branches, provides essential taxonomic information, revealing evolutionary links through their relative positions.

A crucial element in the accurate depiction of phylogenetic trees is the incorporation of branch lengths, representing evolutionary time or genetic change among species or sequences. These lengths contribute significantly to the understanding of the temporal aspects of evolution, offering insights into the processes shaping the \textit{Tree of Life}. Branch lengths play multiple roles across diverse domains, from assessing phylogenetic diversity to identifying and analyzing selection processes \cite{felsenstein1985phylogenies,hahn2005estimating,kosakovsky2005not, volz2013viral,lefort2015fastme,rannala2015art}. Despite substantial advancements in reconstructing phylogenetic trees, the challenge of constructing consensus trees with more topological information and meaningful branch lengths remains an active research area.

In response to this problem, our study introduces an efficient methodology for constructing Maximum Spanning Trees (MST) that consider edge and clade frequencies. This unified framework aims to integrate information from diverse phylogenetic inference methods and data sources, culminating in a consensus tree encapsulating widely supported branches and their associated branch lengths. Our approach promises a refined and comprehensive perspective on evolutionary relationships, shedding light on the consensus time scale of evolution.

\subsection{Our contributions}
\label{sec:contribution}
The following outlines our primary contributions to this research:
\begin{itemize}
    \item We extended PrimConsTree (version 1) from \cite{sifat2024new}, a graph-based method for constructing consensus trees. This method is designed to produce well resolved consensus tree incorporating branch length. This data being highly informative, deserves more attention in the fields of phylogenetic analysis, evolutionary biology, and comparative genomics.

    \item We incorporated a clustering step into the pipeline of PrimConsTree (version 2). Creating homogeneous clusters enables the production of multiple, yet more informative, consensus trees, each of which reveals an alternative gene history.
    
    \item We explored edge frequency and clade frequency estimation in the construction of a MST. These criteria preserve, to the greatest extent possible, the topological aspects of the input trees.

    \item We used our version of PrimConsTree to infer horizontal gene transfer on a group of Archaeabacteria and compare our results with the results obtained by the extended majority rule consensus tree.
\end{itemize}

\section{Definitions and notation}
\label{sec:def}
This article uses standard \emph{phylogenetic tree} terminology. A phylogenetic tree is a rooted, directed tree where every internal node has at least two children and each leaf has a distinct label. Given a tree $T$, the set of all nodes in $T$ is $V(T)$ and it includes two disjoint subsets: the leaf nodes $L(T)$, and the internal nodes (including the root node) $I(T)$. The set of edges in $T$ is $E(T)$ and an edge $(u, v) \in E(T)$ denotes a directed link from node $u$ to $v$. An edge $(u, v)$ represents a parent-child relation where $u$ is the parent of $v$ and $v$ is the child of $u$. The length of an edge $(u, v)$ in $T$ is denoted by $dist_T(u, v)$ and represents the amount of genetic change from $u$ to $v$. Given a node $u \in V(T)$, $T[u]$ means the subtree of $T$ rooted at $u$.

In a phylogenetic tree, a \emph{clade} is any subset of the leaves that have a common ancestor such that no other leaves in the tree have that same node as an ancestor. Given a tree $T$, each node $u \in V(T)$ represents a distinct clade that is defined as $C_u = L(T[u])$. A clade $C_u$ is said to \emph{be supported} by $T$ if there exists a $v \in V(T)$ with $L(T[v])=C_u$. Two distinct clades $C_u$ and $C_v$ are said to be \emph{compatible} either if $C_u \subseteq C_v$, $C_v \subseteq C_u$ or $C_u \cap C_v = \emptyset$.

The set of input trees $S = \{T_1, \dots , T_k\}$ is a set of $k$ phylogenetic trees, all sharing the same set of leaves $L(S) = L(T_1) = \dots = L(T_k)$. The parameters used to measure the size of the input are $k = |S|$ for the number of trees and $n = |L(S)|$ for the number of leaves, respectively. We refer to the subset of trees in $S$ that support the clade $C_u$ by $S_u$, i.e., for every internal node we define $S_u = \{T \in S: u \in V(T)\}$.

In this study, a graph refers to an undirected weighted graph. Given a graph $G$, its set of vertices is $V(G)$, its set of edges is $E(G)$, and its weights are $W(G)$. An edge $(u, v) \in E(G)$ denotes an undirected link between vertices $u$ and $v$ so $(u, v)$ is the same as $(v, u)$. The weight of the edge $(u, v)$ in $G$ is $W_G(u, v)$; as only one graph is involved, we simply write $W(u, v)$ for convenience. Given a subset of vertices $X \subseteq V(G)$, let $G[X]$ denote the subgraph induced by $X$.

\begin{definition}[Consensus tree]
    Let $S = \{ T_1, \ldots, T_k \}$ be a set of $k$ phylogenetic trees on the same set of species $L = L(T_1) = \dots = L(T_k)$. A \emph{consensus tree} of $S$ is a phylogenetic tree $T_c$ with $L(T_c) = L$ that summarizes all of the trees in $S$.
\end{definition}

The main challenge considered in this paper is to find a consensus tree $T_c$ of $S$ that accurately represents the topology as well as the branch lengths of all the trees in $S$.
In the case of building a consensus tree, the objective function ($OF$) of the method can be defined as follows:

\begin{equation}
\label{eq:objective}
    OF = \sum_{i=1}^k dist(T_i, T_c),
\end{equation}

where $k$ is the number of input trees, $T_c$ is the consensus tree, and $dist(T_i, T_c)$ is a distance metric between input phylogenetic tree $i$ (denoted $T_i$) and $T_c$. The objective is to determine the consensus tree, $T_c$, that minimizes the sum of distances across all input trees, thereby optimizing the agreement between $T_c$ and the input set. This optimization seeks to capture the central tendency of the phylogenetic relationships encoded within the input trees while minimizing discordance.

\section{Related work}
\label{sec:work}

Given a set of phylogenetic trees, many methods exist for defining a consensus tree \cite{bryant2003classification,jansson2016improved}. The most well-known types of consensus trees are the strict consensus tree \cite{sokal1981taxonomic,wilkinson2001efficiency}, the majority-rule consensus tree \cite{wilkinson1996majority,margush1981consensusn}, the extended majority consensus tree (also referred to in the literature as the greedy consensus tree) \cite{bryant2003classification,jansson2016improved}, and the frequency
difference consensus tree~\cite{goloboff2003improvements,G05,STACS2024} (called the plurality
consensus tree in~\cite{V_19}). 
All consensus tree inference methods are based on the topology of the input phylogenetic trees. They mostly focus on the representation of each clade in the input trees.
Given $|S_u|$ the number of trees that support the clade $C_u$, the consensus methods are defined as follows:
    
\begin{itemize}
    \item The \emph{strict consensus tree} only retains clades that are supported by all trees, in other words, all the clades $C_u$ such that $|S_u| = k$. This often results in incomplete trees.
    \item The \emph{majority-rule consensus tree} contains all clades that are supported by more than half of the trees ($|S_u| > k/2$), and results in a partially resolved tree.
    \item The \emph{extended majority consensus tree} keeps the majority clades and adds other compatible clades in decreasing order of $|S_u|$. This is the most common approach, leading to commonly more fully resolved trees.
    \item The \emph{frequency difference consensus tree} contains every clade~$C_u$ that is supported by more trees than each of the clades that is incompatible with~$C_u$.
\end{itemize}

Unfortunately, none of the methods described above is able to reconstruct branch lengths. Recently, a novel method for generating a consensus tree that incorporates branch lengths was proposed by Sifat and Tahiri \cite{sifat2024new}. The experimental analysis conducted was limited, involving the utilization of branch length and edge frequency to derive the MST. The obtained results were then compared with a majority-rule consensus tree, revealing a close similarity.

The new method proposed in this article is an extension of the method of Sifat and Tahiri \cite{sifat2024new}, adding key criteria to improve the accuracy and interpretability of the consensus tree. Specifically, we have refined the method with two major improvements. Firstly, we have added a very important clustering phase. This phase reduces conflicts in the input trees and allows for the generation of multiple consensus trees, each with its own specificity, highlighting the different alternative evolutionary histories, increasingly highlighted in the literature \cite{makarenkov2023inferring}. Secondly, we modified the spanning tree criteria to obtain a more biologically relevant consensus tree. While branch length is an important feature, it does not always provide reliable information about the topology of the tree and should therefore not be the sole criterion for deriving the MST. Although some researchers might place greater trust in longer internal branches, as they often correlate with higher clade support, shorter branches in multiple trees can sometimes indicate weak or arbitrary resolutions. For this reason, we prioritized a spanning tree that maximizes edge and clade frequency, which we believe better captures the overall phylogenetic signal. However, we recognize that opinions on this may differ based on specific contexts.

\section{Method}
\label{sec:method}

In this study, we address the complex problem of integrating simultaneously edge frequency and branch length into the construction of a consensus tree. Our method pipeline, depicted in Figure \ref{fig1}, involves a series of steps to derive a consensus tree from the input phylogenetic trees. We subdivide the new method in two main phases, in addition to a preprocessing phase. In the preprocessing phase, we partition the input trees into several clusters (see Section \ref{sec:cluster}), then each cluster can be processed independently by the following main phases. In the first phase, we transform a set of trees into a super-graph. In the second phase, we use an adapted version of the algorithm of Prim to construct a Maximum Spanning Tree (MST) of the super-graph. This MST serves as a base to create the consensus tree. We describe the first phase in Algorithm \ref{alg:supergraph}, and the second phase in Algorithms \ref{alg:prims} and \ref{alg:consensus}. Figure \ref{fig1} (a,b) presents the clustering phase, Figure \ref{fig1} (c) presents the first phase, and Figure \ref{fig1} (d,e) presents the second phase of our method.

\begin{figure*}[!]
\centering
\includegraphics[width=1\textwidth]{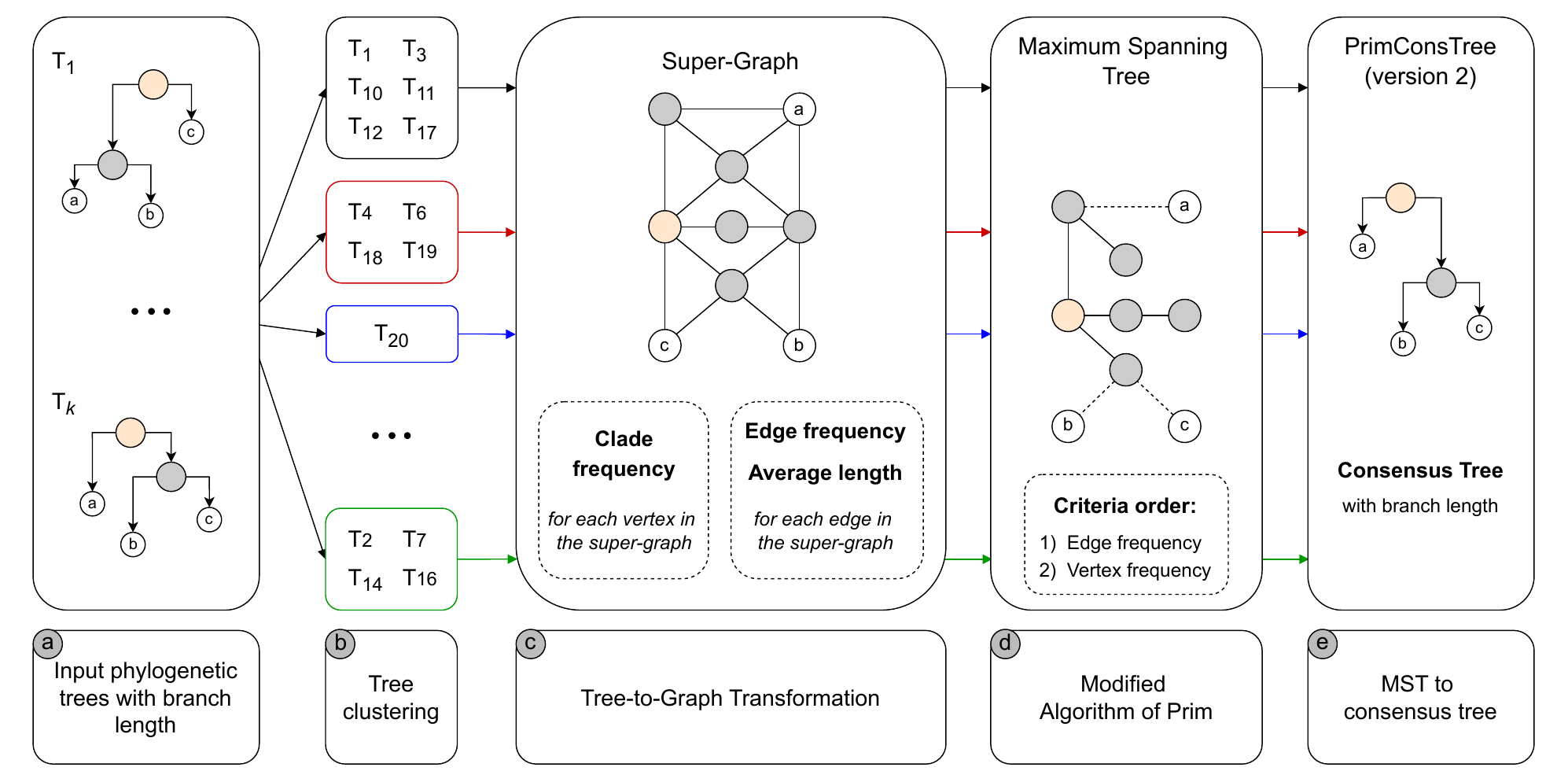}
\caption{Comprehensive visualization of the system architecture overview, illustrating the three main steps of the approach (a-e). The process initiates with the initial input trees and concludes with the generation of one consensus tree per cluster, through the modified Prim algorithm \cite{jarnik1930jistem}. From step (b), each color of the arrows represents an independent instance of the pipeline.}
\label{fig1}
\end{figure*}

%Sub section-------------------1------------------ of Mehtod
\subsection{Clustering}
\label{sec:cluster}
Phylogenetic analysis involves three distinct steps. In the first step, researchers collect data, such as genomic, proteomic, and metabolomic data, for the different taxa under study (e.g., genes, species, morphology). The next step is to apply a tree reconstruction method to the collected data. Many of these methods produce several potential trees for the given data set. Often, hundreds or thousands of trees can be obtained. In the final step, a consensus tree is calculated from the candidate trees to reconcile conflicts, summarize information, and mitigate the large number of potential solutions in the evolutionary story. Although many consensus tree methods exist (see Section \ref{sec:work}), they generally generate a single tree, which poses problems such as loss of information and susceptibility to outliers.

Given $k$ genes defined across $n$ species, the problem is to identify the optimal partition of phylogenetic trees that exhibit similar patterns of evolutionary history while accounting for outliers. 
This intermediate step in constructing a consensus tree facilitates the resolution of conflicts among trees and emphasizes the potential for various alternative consensus trees \cite{tahiri2018new,tahiri2022building}. Figure \ref{fig1} (b) illustrates the clustering step in the process. We have added the choice of $k$-medoids to create a homogeneous clustering with the Silhouette (SH) index \cite{rousseeuw1987silhouettes}. We selected the $k$-medoids clustering method for its ability to handle variability in phylogenetic tree topologies and minimize the influence of outliers. This method optimizes the SH index \cite{rousseeuw1987silhouettes} to form clusters with well-supported topological similarities. Its non-hierarchical nature avoids biases that may arise in hierarchical approaches, making it suitable for datasets with complex evolutionary patterns. The SH index is a method used to interpret and validate the consistency within clusters of data, providing a concise graphical representation of how well each object has been classified. It measures how similar an object is to its cluster (cohesion) compared to other clusters (separation). This value ranges from -1 to +1, with a high value indicating that the object is well-matched to its cluster and poorly matched to neighboring clusters. If most objects have high SH indices, the clustering configuration is considered appropriate. Conversely, if many points have low or negative values, the clustering configuration may have too many or too few clusters.

All subsequent steps following clustering will independently process the different clusters obtained by $k$-medoids (see Figure \ref{fig1}, which shows the various colored arrows). From now on, the set of input trees $S$ will refer to the set of trees within a given cluster.

%Sub section-------------------2------------------ of Mehtod
\subsection{Tree-to-graph transformation}

In this phase, depicted in Figure \ref{fig1} (c), we present a dynamic tree-to-graph transformation, enabling the conversion of a set of input phylogenetic trees into an undirected weighted graph called the super-graph. We explain the annotated naming of internal nodes followed by the construction of the super-graph. During this phase, we also detail the computation of edge frequency, clade frequency, and average edge lengths. All steps described here are processed in Algorithm \ref{alg:supergraph}.

\begin{definition}[Annotated tree]
\label{def:tree_annotated}
An annotated tree $T$ is a tree structure that satisfies the properties of a phylogenetic tree. Additionally, every internal node $u \in I_{T}$ is named according to the process described below.
\end{definition}

The super-graph construction brings together nodes and edges from multiple trees. To identify the occurrences of the same internal node in distinct trees, we transform each phylogenetic tree of $S$ into an annotated tree by identifying its internal nodes. Each internal node is named after its corresponding clade, and any two nodes $u \in I(T_i)$ and $v \in I(T_j)$ are considered equivalent if $C_u = C_v$. Therefore, the same node occurring in multiple trees can reconcile different branching patterns. For example, consider how internal nodes are labeled in Figure \ref{fig:illustr} (a). The node $abc$ appears in two different trees, each time with a distinct underlying topology. It is worth noting that all trees in $S$ share the same set of leaves, thus leaving every root node with the same name after the name attribution of the internal nodes.

It follows that after this step, edges are also identified according to the new names of internal nodes. Two edges $(u, v) \in E(T)$ and $(u', v') \in E(T')$ are equivalent if $u = u'$ and $v = v'$. We do not consider edge lengths or edge directions since the graph is undirected. On the other hand, the edges are distinct if $u \neq u'$ or $v \neq v'$.

\begin{definition}[Annotated trees union]
\label{def:tree_union}
Given a set of annotated trees $D$, the annotated tree union operation is formally denoted $\bigcup_{T_i \in D} T_i$. The result of the operation is an undirected unweighted graph $G$ whose vertex set $V(G)$ can be written as $V(G) = I(G) \cup L(G)$, where $I(G)$ are the
internal nodes of the trees in $D$ and $L(G)$ the leaf nodes of the trees in
$D$, respectively. (For clarity, we reserve the term \textit{vertices} for graphs and the term
\textit{nodes} for trees.) Vertices are given by $I(G) = \bigcup_{T_i \in D} I(T_i)$ and $L(G) = \bigcup_{T_i \in D} L(T_i)$, and edges are given by $E(G) = \bigcup_{T_i \in D} E(T_i)$.
\end{definition}

The super-graph $\mathcal{G}$ is the result of the annotated tree union operation applied on the set of input trees $S$. Every tree shares the same root, so we call the unique vertex corresponding to the root in $V(\mathcal{G})$ the \textit{root vertex} for convenience. Additionally, an \textit{edge frequency} is attached to each edge $(u, v) \in E(\mathcal{G})$ as the edge weight $W(u, v)$. Given an edge $(u, v)$, its frequency $W(u, v)$ is the number of input trees that contain this edge. Formally, let $S_{(u, v)} = \{T \in S: (u, v) \in E(T)\}$ denote the set of trees that contain the edge $(u, v)$. Then, the formula for edge frequency is defined as follows:

\begin{equation}
\label{eq:edge_frequency}
W(u, v) = |S_{(u, v)}|.
\end{equation}

Similarly, a \textit{clade frequency} is assigned to each internal vertex in the super-graph. Let $u$ be an internal vertex in $I(\mathcal{G})$, we note $F(u)$ the clade frequency associated $u$ and give the following formal definition. Given $S_u = \{T \in S: u \in V(T)\}$ the set of trees support the clade $C_u$, and $r$ the root vertex, the formula for clade frequency is defined as follows:

\begin{equation}
\label{eq:vertex_frequency}
F(u) = 
\begin{cases}
    |S_u|,& \text{if } u \neq r.\\
    0,    & \text{otherwise.}
\end{cases}
\end{equation}

Clade frequency provides crucial insight into how often a clade is supported in the input trees. The clades of a phylogenetic tree are highly significant because, collectively, they describe the entire topology of the tree. 
Consequently, clade frequency constitutes an essential metric in the construction of consensus trees, providing an objective measure for representing the recurring topologies across the input trees.

Originally, the clade frequency of the root vertex should be equal to $|S|$ because the corresponding node is present in all trees. We deliberately put it to $0$ to avoid a biased situation where the root vertex gains undue importance merely due to its presence in every tree. This approach allows for a more balanced and accurate representation of clade significance across the super-graph.

Edge and clade frequency are closely related to each other but regardless, each gives important insight on the topology of the input trees. While the presence of an internal node $u$ in a tree relates to the presence of a clade, it does not provide details on the topology of the underlying subtree $T[u]$. On the other hand, because every internal node is named, each edge depicts an explicit link between two nodes, which is much more accurate and therefore, should be given more attention.
However, clade frequency is still very important when tree topologies are too conflicted to refer to explicit edges. For instance, in Figure \ref{fig:illustr} (a), leaf $a$ is never connected with the same edge; nevertheless, the clade $abc$ is represented in two trees out of three. Considering only edge frequency, $a$ might be connected to the root in Figure \ref{fig:illustr} (d), leaving clade $abc$ unresolved in Figure \ref{fig:illustr} (e). The clade frequency suggests a preference for the edge $(abc, a)$ as it contributes to including the clade $abc$.

\begin{figure*}[!]
\centering
\includegraphics[width=1\textwidth]{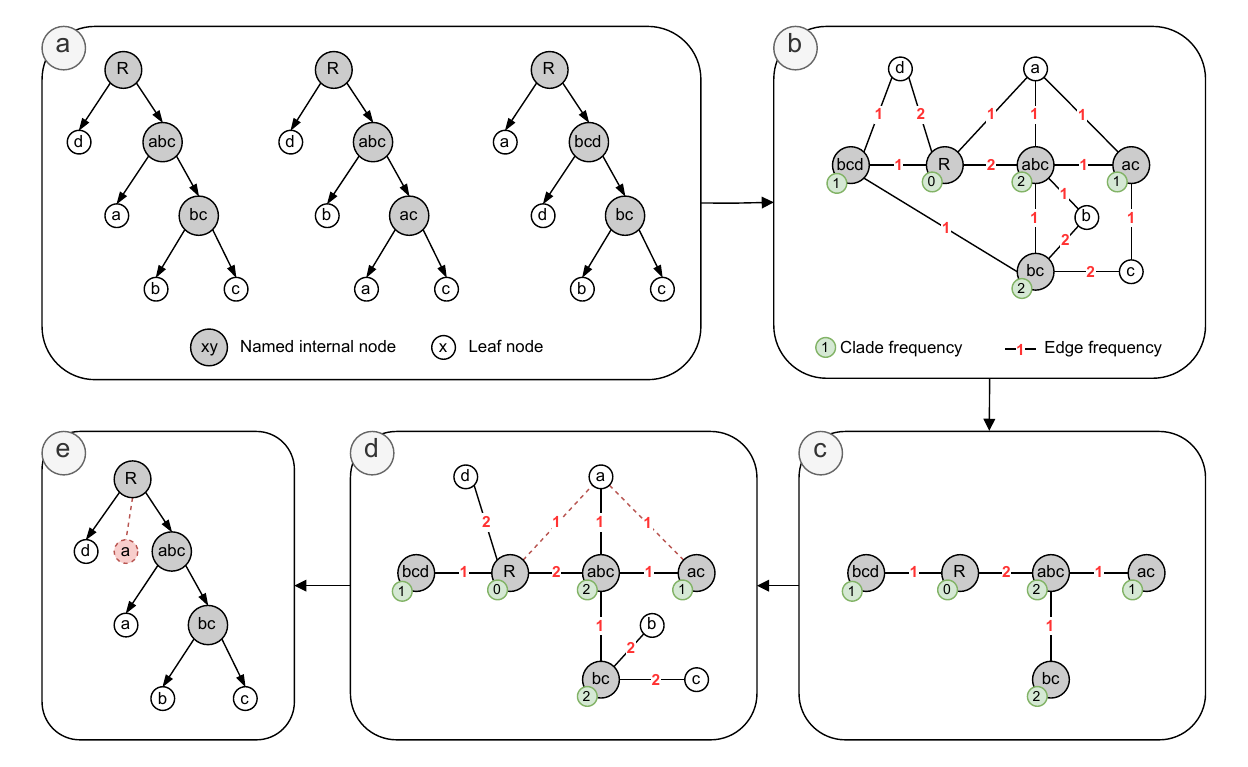}
\caption{This small example outlines how edge and clade frequencies were utilized to retain topological information in PrimConsTree. (a) Three input trees, with already named internal nodes (with root $abcd$ as R); (b) Construct a super-graph encompassing all nodes, incorporating edge and clade frequencies; (c) Derive the Maximum Spanning Tree (MST) from the super-graph, with a focus on internal vertices; (d) Establish connections between leaves and internal vertices based on edge and clade frequencies; (e) The consensus tree is obtained by PrimConsTree by removing unnecessary internal nodes from the MST. In (d, e), red edges outline how ignoring clade frequency could lead to the loss of an important information on the clade $abc$.}
\label{fig:illustr}
\end{figure*}

The last attribute to be attached to the super-graph is \textit{average edge length}. It describes the average length of a given edge across every tree of $S$ that includes the edge. Other trees are omitted because they do not hold relevant information about the length of the edge. The formula to compute average edge length is defined as follows:

\begin{equation}
\label{eq:avg_edge_len}
D(u,v) = \frac{\sum_{T \in S_{(u, v)}} dist_{T}(u, v)}{W(u, v)},
\end{equation}

where $S_{(u, v)} \subseteq S$ is the subset of trees containing an edge between $u$ and $v$, i.e., $S_{(u, v)} = \{T \in S | (u, v) \in E_{T}\}$, $dist_{T}(u, v)$ is the length of this edge in $T$, and $W(u, v)$ is the edge frequency of the edge $(u, v)$ defined in Equation (\ref{eq:edge_frequency}). 

In the context of our analysis, the potential impact of outliers is minimized by the clustering. Therefore, every tree should have its importance in the result and the average is the most appropriate metric because it takes every tree into account.

\begin{lemma}[Connectivity of the super-graph]
\label{lem:connect_supergraph}
Let $\mathcal{G}$ be the super-graph resulting from the annotated trees union operation, applied on the set of input trees $S$. The graph $\mathcal{G}$ is connected and the subgraph $\mathcal{G}[I(\mathcal{G})]$ restricted to internal vertices is also connected.
\end{lemma}

\begin{proof}
Every tree in $S$ has the same set of leaves $L(S)$, so they all share the same root. Additionally, every tree $T \in S$ is a tree structure, so it holds paths from the root to every other node. Those paths still exist in $\mathcal{G}$ because every distinct edge is included in $E(\mathcal{G})$. Therefore, in $\mathcal{G}$, there is a path from the root vertex to any other vertex in $V(\mathcal{G})$, then $\mathcal{G}$ is connected. Moreover, by definition, a leaf in a tree is connected to exactly one internal node, so it cannot lie on a path between the root and an internal node. Then, in every tree, any path that connects the root with an internal node only contains internal nodes. Consequently, those paths still present in $\mathcal{G}[I(\mathcal{G})]$ and $\mathcal{G}[I(\mathcal{G})]$ are connected.
\end{proof}

Algorithm \ref{alg:supergraph} describes the process used to build the super-graph. Initially, the super-graph $\mathcal{G}$ is empty, first, we include the set of leaves such that $L(\mathcal{G}) = L(S)$, and then, we sequentially incorporate information from each tree $T \in S$. When adding a tree $T$, the graph internal vertices are updated as $I(\mathcal{G}) = I(\mathcal{G}) \cup I(T)$ and the graph edges are updated as $E(\mathcal{G}) = E(\mathcal{G}) \cup E(T)$. For each edge $(u, v) \in E(T)$, the edge frequency is increased by one and the average length is incremented by $dist_T(u, v)$. The clade frequency $F(u)$ of each node $u \in T$ is also increased by one. Finally, the average length of each edge is divided by its frequency.

\begin{algorithm}[!t]
\caption{\textsc{SuperGraph}: builds a super-graph gathering information from a set of input trees}
\label{alg:supergraph}
\begin{algorithmic}
\State \textbf{Input:} $S = \{ T_1,\cdots, T_k \}$ a set of $k$ phylogenetic trees \\
\hspace{0.9cm} sharing the same set of leaves $L(S)$
\State \textbf{Output:} $\mathcal{G}$ the super-graph including: \\
\hspace{2.1cm} - internal vertices $I(\mathcal{G})$, \\
\hspace{2.1cm} - leaf vertices $L(\mathcal{G})$, \\
\hspace{2.1cm} - edge frequencies $W(\mathcal{G})$, \\
\hspace{2.1cm} - clade frequencies $F$, and \\
\hspace{2.1cm} - average edge lengths $D$
\end{algorithmic}
\begin{algorithmic}[1]
\Function{SuperGraph}{$S$}
    \State $\mathcal{G} \gets \langle V(\mathcal{G}), E(\mathcal{G}), W(\mathcal{G}) \rangle$ \Comment{Empty super-graph}
    \State $L(\mathcal{G}) \gets L(S)$
    \For{$i \gets 1$ to $k$} \Comment{Incorporate $T_i$ into $\mathcal{G}$}
        \State Rename all internal nodes in $T_i$
        \State $I(\mathcal{G}) \gets I(\mathcal{G}) \cup I(T_i)$
        \State $E(\mathcal{G}) \gets E(\mathcal{G}) \cup E(T_i)$
        \ForAll{edge $(u, v) \in E_{T_i}$}
            \State $W(u, v) \gets W(u, v) + 1$ \Comment{Eq. (\ref{eq:edge_frequency})}
            \State $F(v) \gets F(v) + 1$ \Comment{Eq. (\ref{eq:vertex_frequency})}
            \State $D(u, v) \gets D(u, v) + dist_{T_i}(u, v)$
        \EndFor
    \EndFor
    \ForAll{edge $(u, v) \in E(\mathcal{G})$} \Comment{Average edge length}
        \State $D(u, v)  \gets D(u, v) / W(u, v)$ \Comment{Eq. (\ref{eq:avg_edge_len})}
    \EndFor
    \State \Return $\mathcal{G}$
\EndFunction
\end{algorithmic}
\end{algorithm}

At the end of the first phase, the main properties of the input trees are retained in the super-graph:
\begin{itemize}
    \item Tree topologies are retained firstly in the graph topology and secondly in its attributes, including edge frequencies $W$ and clade frequencies $F$. 
    \item Branch length data is retained in the average edge length attribute $D$.
\end{itemize}

%Sub section-------------------3------------------ of Mehtod
\subsection{Consensus tree construction from the super-graph}

In the following phase, depicted in Figure \ref{fig1} (d,e), the analysis involves the super-graph $\mathcal{G}$, including edge and clade frequency and average edge length. 
Algorithm \ref{alg:prims} builds a specific Maximum Spanning Tree (MST) of the super-graph and Algorithm \ref{alg:consensus} uses the MST as a base to generate a consensus tree $T_c$. Indeed, an MST is a straightforward and efficient way to find a tree structure that maximizes the values of its attributes (i.e., edge and clade frequency), consequently maximizing the accordance of its topology with the input trees. Another advantage of the MST is that it directly takes edges from the super-graph, which already retains branch length data.

In the context of a consensus tree, the resulting consensus tree must have the same set of leaves as the phylogenetic trees used as input. To ensure that the MST does not compromise this condition, we divide its construction into two main steps. The first step (see Algorithm \ref{alg:prims} lines 1 to 17) utilizes the subgraph $\mathcal{G}[I(\mathcal{G})]$ restricted to internal vertices only, and yields its Maximum Spanning Tree (MST). Indeed, constructing the MST from all vertices could lead to the conversion of leaf vertices into internal nodes within the MST, a scenario deemed undesirable. In the second step, the vertices in $L(\mathcal{G})$ are connected to the MST using the function \textsc{AttachLeaves} depicted in Algorithm \ref{alg:primleaves}.

The process of Algorithm \ref{alg:prims} follows the course of the original algorithm of Prim. It starts by including the root vertex in the MST. Then, at each iteration, it looks for every edge that has exactly one endpoint included in the MST and includes the one of higher weight. Subsequently, the process repeats until all vertices are included in the MST. The particularities of our modified Algorithm of Prim lie in two points. Firstly, the set of leaf vertices $L(\mathcal{G})$  is connected independently after the rest of the vertices. Secondly, in addition to the weight (i.e., edge frequency), the clade frequency is used to choose the optimal edge. 
 
Let $(u, v)$ be a candidate edge at a given iteration, where $u \in I(\mathcal{G})$ is included in the MST, $v \in I(\mathcal{G})$ is not included in the MST and $(u, v) \in E(\mathcal{G})$.
To choose the optimal edge, Algorithm \ref{alg:prims} first finds the edge of highest edge frequency $W(u, v)$. Indeed, edges give the most precise insights into the topology so edge frequency must be considered in priority.

In case multiple candidates of maximal edge frequency are available, clade frequencies of the edge endpoints are taken into account. First, the algorithm considers the clade frequency of the vertex that is not included in the MST, in other words, it maximizes $F(v)$. Then, if multiple optimal edges are still available, the clade frequency of the vertex included in the MST, $F(u)$, is maximized.

Figure \ref{fig:MST} shows how each criterion is used to find the MST. Starting from vertex $a$, the edge $(a, b)$ is chosen according to its edge frequency. Secondly, the edge $(b, c)$ is chosen over $(b, d)$ because of the clade frequency $F(c) = 2$ (i.e., clade frequency of the vertex not already included). Finally, edge $(b, d)$ is chosen over $(c, d)$ because of the clade frequency $F(b) = 3$ (i.e., clade frequency of the vertex already included).

\begin{figure}[!]
\centering
\includegraphics[width=0.2\textwidth]{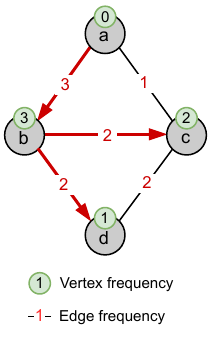}
\caption{Illustration of our modified algorithm of Prim (i.e., PrimConsTree version 2). Vertex names are defined arbitrarily. The process starts from root vertex $a$ and makes use of each criterion.}
\label{fig:MST}
\end{figure}

Then, Algorithm \ref{alg:primleaves} is used to connect each leaf vertex to the most suitable internal node. This is done in a very similar way to the rest of the MST. Given an internal vertex $u \in I(\mathcal{G})$ and a leaf vertex $v \in L(\mathcal{G})$ the criteria are maximized in the following order: first edge frequency $W(u, v)$, and then clade frequency of the internal vertex $F(u)$. Algorithm \ref{alg:primleaves} is depicted in the appendix as it is very close to the MST of internal vertices.

\begin{algorithm}[!t]
\caption{\textsc{ModifiedPrim}: build a maximum spanning tree with specific criteria and a predefined set of leaves}
\label{alg:prims}
\begin{algorithmic}
\State \textbf{Input:} $\mathcal{G}$ the super-graph including: \\
\hspace{1.7cm} - internal vertices $I(\mathcal{G})$, \\
\hspace{1.7cm} - leaf vertices $L(\mathcal{G})$, \\
\hspace{1.7cm} - edge frequencies $W(\mathcal{G})$, and \\
\hspace{1.7cm} - clade frequencies $F$
\State \textbf{Output:} $\mathcal{M}$ a maximum spanning tree on $\mathcal{G}$ \\
\hspace{1.2cm} such that $L(\mathcal{G}) \subseteq L(\mathcal{M})$
\end{algorithmic}
\begin{algorithmic}[1]
\Function{ModifiedPrim}{$\mathcal{G}$}
    \State $\mathcal{M} \gets \langle V(\mathcal{M}), E(\mathcal{M}) \rangle$ \Comment{Empty MST}
    \State Add the root vertex in $V(\mathcal{M})$
    \While{$|V(\mathcal{M})| < |I(\mathcal{G})|$}
        \State $e^* \gets (\textbf{Null}, \textbf{Null})$ \Comment{Optimal edge to add}
        \State $C \gets \{(u, v) \in E(\mathcal{G}): u \in V(\mathcal{M}),$  
        \State \hspace{9.21em} $v \notin V(\mathcal{M}),$
        \State \hspace{9.21em} $v \in I(\mathcal{G})\}$
        \ForAll{edge $(u, v) \in C$}
            \If{$W(u, v) > W(e^*)$}
                \State $e^* \gets (u, v)$
            \ElsIf{$W(u, v) = W(e^*)$}
                \If{$F(v) > F(e^*[1])$}
                    \State $e^* \gets (u, v)$
                \ElsIf{$F(v) = F(e^*[1])$}
                    \If{$F(u) > F(e^*[0])$}
                        \State $e^* \gets (u, v)$
                    \EndIf
                \EndIf
            \EndIf
        \EndFor
        \State $V(\mathcal{M}) \gets V(\mathcal{M}) \cup \{e^*[1]\}$
        \State $E(\mathcal{M}) \gets E(\mathcal{M}) \cup \{e^*\}$
    \EndWhile
    \State $\mathcal{M} \gets$ \textsc{AttachLeaves}($\mathcal{M}$, $\mathcal{G}$) \Comment{Alg. (\ref{alg:primleaves})}
    \State \Return $\mathcal{M}$
\EndFunction
\end{algorithmic}
\end{algorithm}

The final step is to use the MST as a base to create a valid consensus tree $T_c$ that respects the following conditions. Firstly, its edges must be directed and must have edge length attached. Secondly, every internal node must have at least two children. Finally, the set of leaves must be equal to the initial set of leaves such that $L(T_c) = L(S)$. The Algorithm \ref{alg:consensus} is used to arrange the MST into a valid consensus tree.

The consensus tree is directed by positioning the root as the root vertex. For each edge $(u, v)$ in the MST, the edge is added $E(T_c)$, pointing away from the root. At the same time, the length $dist_{T_c}(u, v) = D(u, v)$ is attached to the edge. 

Finally, we define two types of extra nodes that must be removed in order to get a proper consensus tree: \textit{unnecessary} internal nodes and \textit{redundant} internal nodes. Any internal node $u \in I(\mathcal{G})$ that became a leaf in the MST, is considered unnecessary and can easily be removed because it does not have children. On the other hand, an internal node is redundant if it has only one child (leaf or internal node). Let $u$ be a redundant node, $p(u)$ is its parent, and $c(u)$ is its unique child. Then the node $u$ is deleted and a new edge $(p(u), c(u))$ is added. The length of the new edge is the total length of the previous path such that $dist_{T_c}(p(u), c(u)) = dist_{T_c}(p(u), u) + dist_{T_c}(u, c(u))$.

At the end of this phase, a proper consensus tree $T_c$ is returned by Algorithm \ref{alg:consensus}. 

We would like to emphasize that the output of PrimConsTree version 2 is not necessarily a binary tree. While the recursive edge selection process might initially suggest a binary structure, the inclusion of non-binary relationships remains possible.
Specifically, during the node removal step (Algorithm \ref{alg:consensus}), internal nodes with multiple children that do not contribute to distinct clades can lead to the emergence of non-binary nodes in the final consensus tree. The node removal process in Algorithm 3 ensures that non-binary structures can be maintained when internal nodes with multiple children fail to contribute uniquely to the topology. This step allows for flexibility, supporting non-binary consensus trees when appropriate, reflecting ambiguous or weakly supported evolutionary relationships. The AttachLeaves (Algorithm \ref{alg:primleaves}) plays a critical role in handling terminal edges for taxa represented by single nodes, ensuring that each taxon is properly attached to the consensus tree. This step is essential in maintaining both binary and non-binary relationships, especially when dealing with taxa that form terminal clades.

\begin{algorithm}[!t]
\caption{\textsc{ConsensusFromMST}: orient the tree, remove unnecessary/redundant internal nodes, and attach branch length}
\label{alg:consensus}
\begin{algorithmic}
\State \textbf{Input:} $\mathcal{M}$ the minimum spanning tree of $\mathcal{G}$ and \\
\hspace{0.9cm} $\mathcal{G}$ the super-graph with vertices $V(\mathcal{G}) = I(\mathcal{G}) \cup L(\mathcal{G})$ \\
\hspace{0.9cm} and average edge length $D$
\State \textbf{Output:} $T_c$ a consensus tree with branch length
\end{algorithmic}
\begin{algorithmic}[1]
\Function{ConsensusFromMST}{$\mathcal{M}$, $\mathcal{G}$}
    \State $T_c \gets \langle V(T_c), E(T_c) \rangle$ \Comment{Initialize empty consensus tree}
    \State $L(T_c), I(T_c) \gets L(\mathcal{G}), I(\mathcal{G})$
    \ForAll{edge $(u, v) \in E(\mathcal{M})$}
        \State Orient $(u, v)$ to point away from the root
        \State Add the edge $(u, v)$ to $E(T_c)$
        \State $dist_{T_c}(u, v) \gets D_{u, v}$ \Comment{Attach average edge length}
    \EndFor
    \ForAll{node $u \in V(T_c)$ in postorder}
        \If{$u \notin L(T_c)$}
            \State $children \gets$ list of children of $u$ in $T$
            \If{$|children| = 0$} %\Comment{Unnecessary $u$}
                \State Remove $u$ from $V_{T}$
            \ElsIf{$|children| = 1$} %\Comment{Redundant internal node}
                \State $p(u) \gets$ parent of $u$ in $mst$
                \State $c(u) \gets children[0]$
                \State Add the edge $(p(u), c(u))$ to $E(T_c)$
                \State $dist_{T_c}(p(u), c(u)) \gets dist_{T_c}(p(u), u)$
                \State \hspace{8.5em} $+ dist_{T_c}(u, c(u))$
                \State Remove $u$ from $V(T_c)$
            \EndIf
        \EndIf
    \EndFor
    \State \Return $T_c$
\EndFunction
\end{algorithmic}
\end{algorithm}

%Sub section-------------------4------------------ of Mehtod
\subsection{Time complexity}

\begin{lemma}
    The number of vertices and edges in the super-graph $\mathcal{G}$ are both
$O(nk)$.
\end{lemma}

\begin{proof}
    Given a binary tree $T$, its number of nodes is $|V(T)| = 2n - 1$ and the number edges is $|E(T)| = 2n - 2$. In the worst case, every tree in $S$ is binary, therefore, at most $k \cdot (2n-1)$ distinct nodes are included in $V(\mathcal{G})$ and at most $k \cdot (2n-2)$ distinct edges are included in $E(\mathcal{G})$.
\end{proof}

\begin{lemma}
    Building the super-graph in Algorithm \ref{alg:supergraph} is done in $O(nk)$ time.
\end{lemma}

\begin{proof}
    The first loop (starting at line 4) iterates $k$ times and all its operations can be done in $O(n)$ time, resulting in an $O(nk)$ time complexity. The second loop (starting at line 14) iterates over the edges of $E(\mathcal{G})$ so it finishes in $O(nk)$ time, too. They are executed sequentially so the time complexity is $O(nk)$.
\end{proof}

\begin{lemma}
    Finding a maximum spanning tree $\mathcal{M}$ in Algorithm \ref{alg:prims} is done in $O(nk \cdot \log(nk))$ time.
\end{lemma}

\begin{proof}
    The MST of internal vertices (lines 1 to 17) add only constant time operations to the algorithm of Prim, so as the original algorithm, it can be implemented in $O((|V|+|E|) \cdot \log(|V|))$ time with a priority queue, which raises a time complexity of $O(nk \cdot \log(nk))$ for $\mathcal{G}$. Then to attach leaf vertices in Algorithm \ref{alg:primleaves}, $n$ leaf vertices are traversed (line 2) and for each, at most $k$ distinct edges can be available, so leaf vertices are attached in $O(nk)$. Consequently, the time complexity of Algorithm \ref{alg:prims} is $O(nk \cdot \log(nk))$.
\end{proof}

\begin{lemma}
In Algorithm \ref{alg:consensus}, the consensus tree $T_c$ is obtained from the MST in $O(nk)$ time.
\end{lemma}

\begin{proof}
    The number of nodes and edges in $\mathcal{M}$ are $O(nk)$. The first loop (starting at line 4) iterates through the edges in $E(\mathcal{M})$, doing only constant-time operations. The second loop line (starting at line 9) iterates over the nodes in $V(T_c)$, which is at this point equal to $V(\mathcal{M})$. Therefore both loops finish in $O(nk)$ time and Algorithm \ref{alg:primleaves} finishes in $O(nk)$ time.
\end{proof} 

\begin{theorem}
The time complexity to build a consensus tree from a set of $k$ phylogenetic trees on a set of $n$ leaves using PrimConsTree is in $O(nk \cdot \log(nk))$.
\end{theorem}

\begin{proof}
    PrimConsTree involves three main algorithms that run sequentially: Algorithm \ref{alg:supergraph} builds the super-graph in $O(nk)$ time, Algorithm \ref{alg:prims} finds an MST in $O(nk \cdot \log(nk))$ time and Algorithm \ref{alg:consensus} yields a proper consensus tree in $O(nk)$ time. The total time complexity of PrimConsTree is $O(nk \cdot \log(nk))$.
\end{proof}

\section{Comparative analysis}

In this section, we will provide a detailed description of the datasets employed—both simulated and biological—that were utilized to assess the performance of our novel PrimConsTree methodology.

\subsection{Simulations}
To generate phylogenetic trees, the HybridSim program \cite{woodhams2016simulating} was utilized with specified parameters: the number of trees $k \in \{10, 30, \dots, 130, 150\}$, the number of leaves $n \in \{10, 20, 30, 40, 50\}$, and the coalescence rate $c \in \{10, 7.5, 5, 2.5, 1\}$. For each combination of these parameters, ten distinct sets of input trees were produced. For each set of input trees, the PrimConsTree version 2 was computed to summarize the phylogenetic relationships represented by the input data. The analysis also included the extended majority rule \cite{bryant2003classification,jansson2016improved} to assess clade support, as well as the frequency difference method \cite{goloboff2003improvements} to evaluate the variability in clade support across the generated trees. Additionally, the previous version of PrimConsTree version 1 \cite{sifat2024new} was employed, which considered minimum branch lengths and edge frequencies in the construction of the consensus tree.

To evaluate the quality of each consensus tree, we measured the average distance between the consensus tree and each input tree, as described in Equation (\ref{eq:objective}). During the simulation step, we utilized the Kendall–Colijn distance \cite{kendall2016mapping}, which is integrated into Equation (\ref{eq:objective}). The formulation of this distance metric is presented as follows:

\begin{definition}[Kendall-Colijn distance \cite{kendall2016mapping}]
\label{def:KCdistance}
The Kendall–Colijn distance quantifies the dissimilarity between two phylogenetic trees $T_1$ and $T_2$ by incorporating both topological and branch length information. It is defined as:
\end{definition}

\begin{equation}
    \label{eq:kcdist}
    D_{KC}^\lambda(T_1, T_2) = \sum_{u,v \in L| u \neq v} \left| (1 - \lambda) \cdot \text{dist}_{\text{top}}(u, v) + \lambda \cdot \text{dist}_{\text{BL}}(u, v) \right|,
\end{equation}

where $\text{dist}_{\text{top}}(u, v) $ is the topological difference between taxa $u$ and $v$ in trees $T_1$ and $T_2$, $\text{dist}_{\text{BL}}(u, v)$ represents the difference in branch lengths between taxa $u$ and $v$ in the two trees, calculated as the absolute difference in the path lengths connecting them in $T_1$ and $T_2$, and $\lambda \in [0, 1]$ is a weighting parameter that balances the contributions of topological and branch length differences, allowing for flexibility in the comparison. A value of $\lambda = 1$ emphasizes branch length differences, while $\lambda = 0$ focuses solely on topological differences. We alternatively employed the Kendall-Colijn distance without and with consideration of branch lengths (respectively $\lambda=0$ and $\lambda=0.5$) in Equation (\ref{eq:kcdist}).

In the simulations using the extended majority rule algorithm implemented in Bio.Phylo (Python v1.75, Consensus module) \cite{cock2009biopython}, branch lengths in the consensus tree are calculated by averaging the branch lengths of corresponding clades across the input trees. In contrast, the frequency difference consensus tree algorithm disregards branch lengths; all branches in the resulting consensus tree are assigned a fixed length of $1$, making this method strictly topological.

\subsection{Biological data}

In our research, we applied a new algorithm to analyze the evolution of 47 ribosomal proteins from 14 Archaeabacteria organisms, including 11 Euryarchaeota species and 3 Crenarchaeota species. These data, initially studied by Matte-Tailliez et al. \cite{matte2002archaeal} then by Boc et al. \cite{boc2010inferring}, led to the inference of a single species phylogenetic tree by concatenating the protein sequences. The evolution of each protein can be visualized through its unique phylogenetic tree. Utilizing cluster analysis on these trees can unveil distinct evolutionary scenarios represented by the sequences.  By grouping similar gene trees together, we can isolate those that exhibit consistent lineage sorting patterns, allowing for a clearer understanding of evolutionary processes at play.

The PHYML method was used to infer 47 phylogenetic trees from the 47 complete alignments. We then applied our $k$-medoids \cite{tahiri2018new} tree clustering algorithm, based on SH \cite{rousseeuw1987silhouettes} index and the non-squared Robinson-Foulds distance, to partition the set of 47 trees. We followed the same protocol as in our previous study \cite{tahiri2018new}, in which the $k$-medoids clustering method was used. By maintaining the same clustering strategy, we ensure comparability of consensus trees generated by different inference methods. In this study, we used the PrimConsTree version 2 algorithm for consensus tree inference, compared to majority-rule consensus trees used in previous work.
The maximum SH index indicated five clusters, corresponding to five different horizontal gene transfer scenarios, illustrated in Figure \ref{fig:results} (b-f).

The first cluster contained 11 trees, the second 4 trees, the third 20 trees, the fourth 11 trees, and the fifth 1 tree. For these clusters, we inferred the extended majority consensus trees and PrimConsTree version 2. The tree reconciliation process, involving subtree pruning and regrafting (SPR) moves (representing horizontal gene transfers) of species tree clusters, resulted in the transformed topology of the species tree matching that of the gene tree. The computations were performed using the version of the algorithm available on the T-Rex website \cite{boc2012t}.

\section{Results}

In this section, we first present the results of our simulations (Section 6.1), which were conducted using synthetic data described in Section 5.1. Subsequently, in Section 6.2, we detail the findings of our clustering analysis on the evolutionary patterns of 47 ribosomal proteins derived from data of 14 Archaeabacteria organisms, originally investigated by Matte-Tailliez et al. \cite{matte2002archaeal}.

\subsection{Simulations}
We conducted a performance evaluation of PrimConsTree, the frequency difference consensus tree, and the extended majority rule consensus tree using the input trees generated for this purpose. The initial version of PrimConsTree, as proposed by Sifat and Tahiri \cite{sifat2024new}, is designated as PrimConsTree version 1, while the version developed in this study is referred to as PrimConsTree version 2.

The evaluation results, illustrated in Figure \ref{fig:results_kc} (d, e, f), are based on minimizing the objective function using the Kendall-Colijn distance with $\lambda = 0.5$, which integrates both topological and branch length information into the distance measure. PrimConsTree version 2 consistently demonstrated superior accuracy compared to the other methods, particularly when the number of input trees exceeded 30. As the number of trees ($k$) increased, the precision of PrimConsTree version 2 continued to improve, as shown in Figure \ref{fig:results_kc} (d). In addition, Figure \ref{fig:results_kc} (e) highlights that PrimConsTree version 2 maintained its performance advantage across a wide range of leaf counts, with especially notable improvements for larger numbers of leaves ($n$).

The increase in distance between the input trees and the consensus tree, observed for all methods as the number of leaves increased, can be attributed to the greater incidence of conflicting phylogenetic signals in larger trees. This reflects the inherent complexity of reconciling trees with a higher degree of topological variance. Moreover, a higher coalescence rate, which indicates reduced genetic divergence between lineages, resulted in consensus trees more closely approximating the input trees. This suggests that higher rates of coalescence may limit the divergence of individual trees, thereby constraining the topological variance across the input set.
In most scenarios evaluated, PrimConsTree version 2 exhibited superior performance relative to PrimConsTree version 1. This outcome aligns with theoretical expectations: while PrimConsTree version 1 selects the most frequent edges across the input trees, PrimConsTree version 2 prioritizes edges with the shortest average branch lengths, which may correspond to more reliable evolutionary signals. The association between short branch lengths and low clade support, as discussed by Wiens \cite{Wiens2008}, is particularly relevant here, as the enhanced method emphasizes edges with more consistent support across trees.

\begin{figure*}[!htp]
    \centering
    \includegraphics[width=\textwidth]{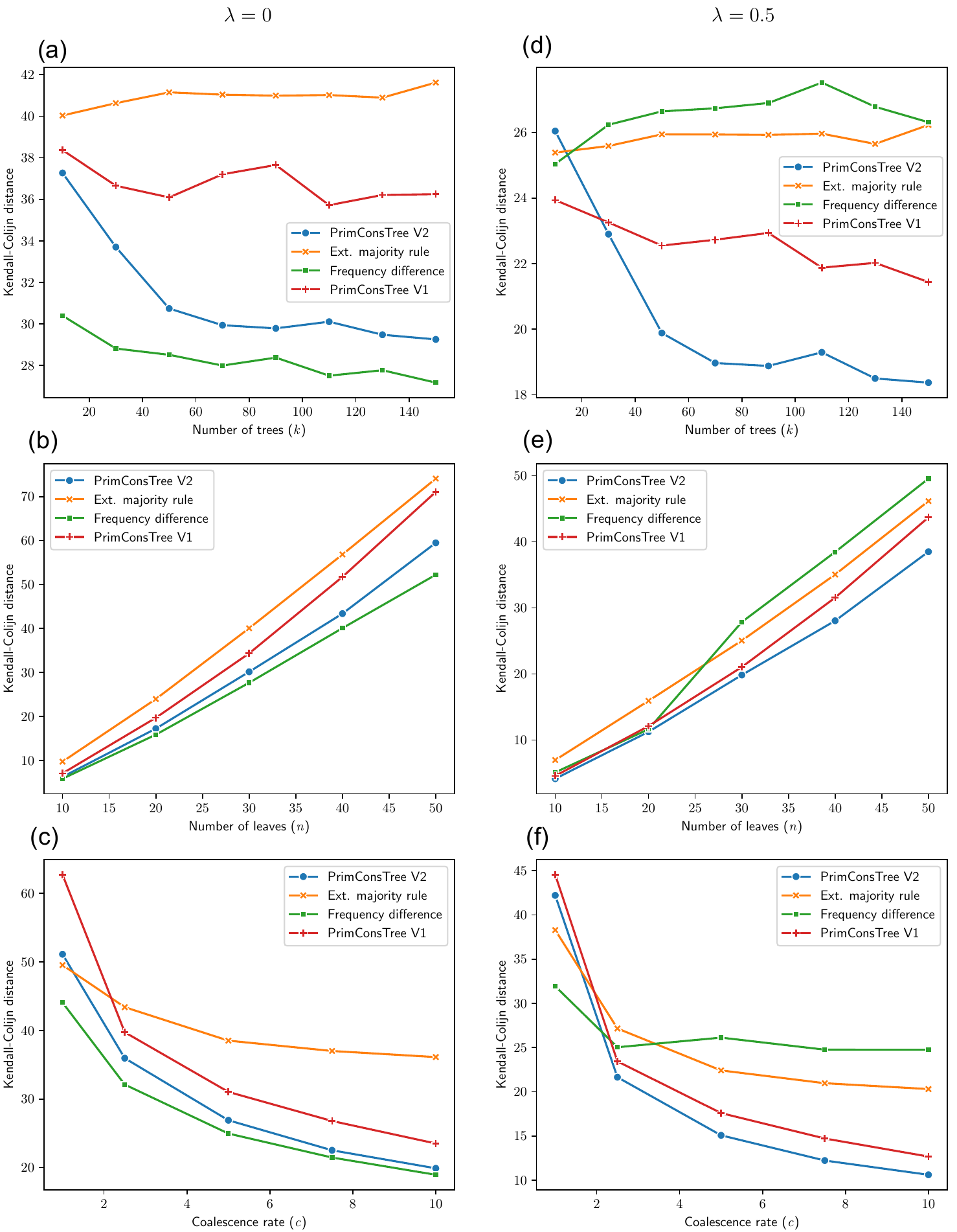}
    \caption{Evaluation of the objective function using the Kendall–Colijn distance for PrimConsTree version 1 (denoted as V1) and version 2 (denoted as V2), as well as for Extended Majority Rule and Frequency Difference consensus trees. In plots (a, b, c), branch lengths are considered ($\lambda = 0.5$), while plots (d, e, f) focus solely on topology ($\lambda = 0$).}
    \label{fig:results_kc}
\end{figure*}

Although branch lengths are integral to the PrimConsTree methodology, the results presented in Figure \ref{fig:results_kc} (a, b, c) focus exclusively on topological accuracy, using the Kendall-Colijn metric with $\lambda = 0$, where only tree topology is considered. The frequency difference consensus tree, which does not incorporate branch lengths, proved to be robust in preserving the underlying topological structure, yet PrimConsTree remains competitive and frequently outperformed this method in terms of overall accuracy, even in purely topological assessments.
To better understand where PrimConsTree is the most relevant method, we show in Figure \ref{fig:results_lamb} the evaluation of the Kendall-Colijn metric with respect to the value of the $\lambda$ parameter.

\begin{figure}[!htp]
    \centering
    \includegraphics[width=0.7\textwidth]{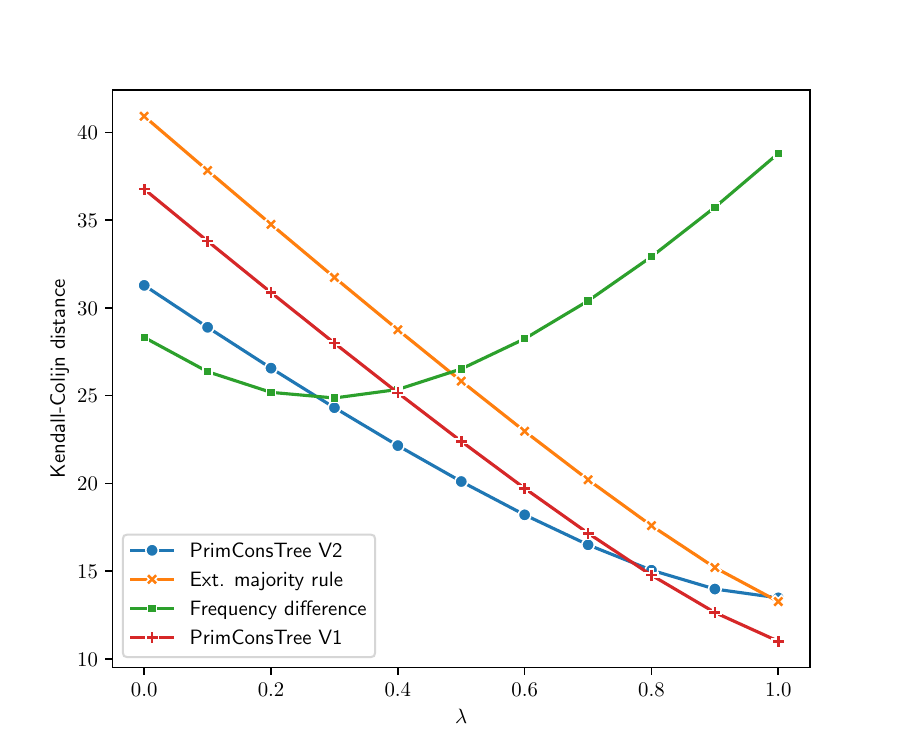}
    \caption{Evaluation of the objective function using the Kendall–Colijn distance for different values of $\lambda$. Methods evaluated are PrimConsTree version 1 (denoted as V1) and version 2 (denoted as V2), as well as for Extended Majority Rule and Frequency Difference consensus trees.}
    \label{fig:results_lamb}
\end{figure}

We carried out a supplementary analysis to assess the performance of the Kendall-Colijn distance formulas (Equation (\ref{eq:objective}) and Equation (\ref{eq:kcdist})) at different $\lambda$ values, comparing the accuracy of PrimConsTree version 1, PrimConsTree version 2, Extended Majority Rule and Frequency Difference consensus trees. The results are shown in Figure \ref{fig:results_kc}. Our method outperformed the other approaches in the $\lambda$ range from 0.3 to nearly 0.8, demonstrating robustness in scenarios where both topology and branch length are evaluated simultaneously. Even beyond this range, PrimConsTree version 2 consistently ranked as the second best method, highlighting its reliability and adaptability in cases where the relative importance of topology and branch length is uncertain or not explicitly defined. These results indicate that our algorithm is particularly effective at balancing these two factors over a wide range of weighting schemes, making it a valuable tool in applications where the contribution of branch length is significant.

\subsection{Biological data}

The results inferred from the analysis, presented in Figure \ref{fig:results} (b to f), account for five different histories. While the clustering of gene trees may suggest the presence of different evolutionary histories, it is important to consider alternative explanations, such as biases due to substitution model misfit, long-branch attraction, or other analytical artifacts. These factors could lead to replicated incongruence across gene trees, particularly when small datasets are used for tree reconstruction. Thus, the observed clusters may reflect a combination of true biological signals, such as horizontal gene transfer (HGT) or recombination, and methodological biases, which should be carefully evaluated.

\begin{figure*}[!]
\centering
\includegraphics[width=\textwidth]{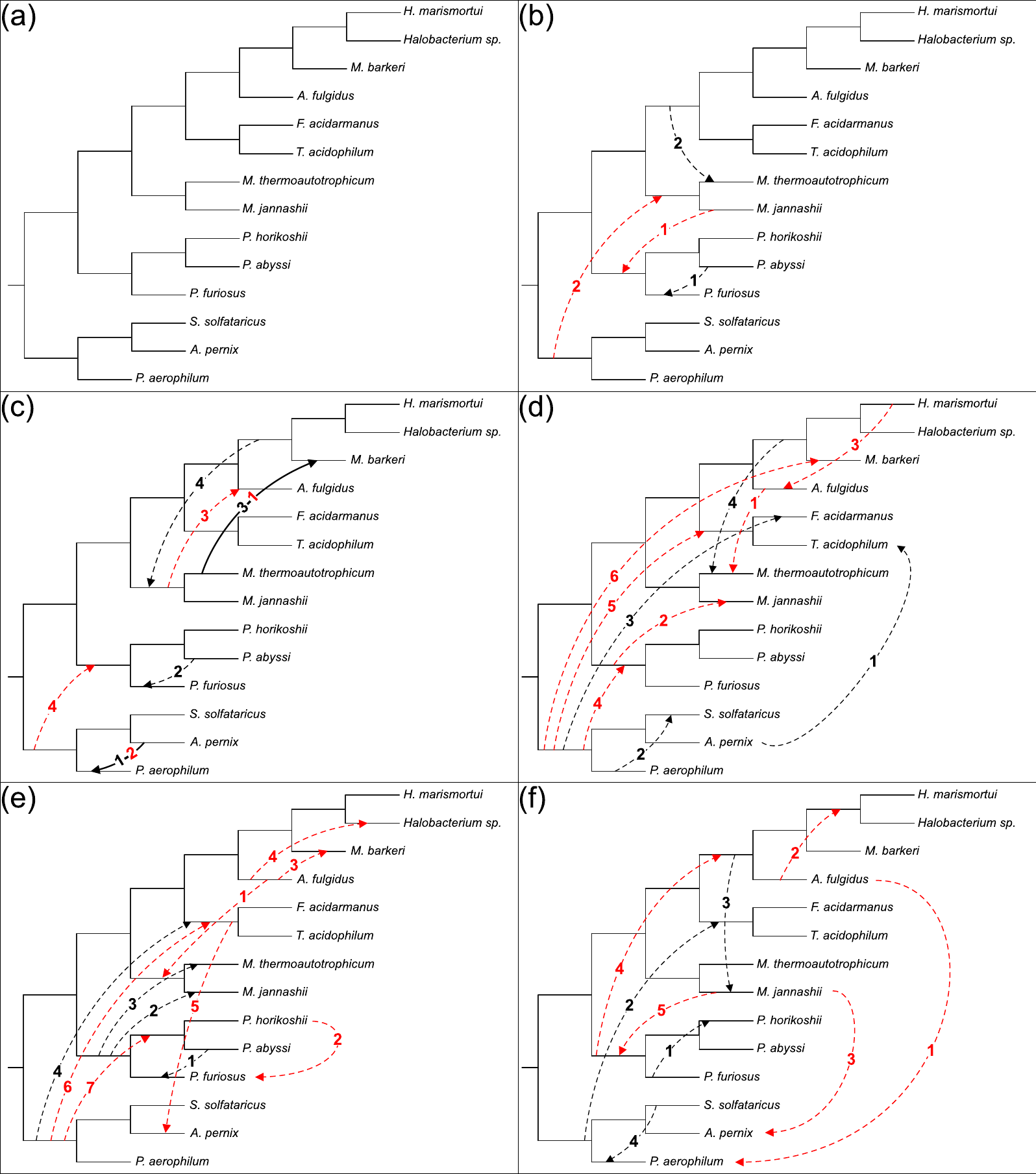}
\caption{This figure, adapted from Figure 6 (a) in \cite{tahiri2018new}, present alternative scenarios of evolution found using the majority consensus trees and PrimConsTree version 2. Panel (a) is the species tree and panels (b, c, d, e, f) are respectively inferred from clusters SH1 to SH5. Transfers are labelled in the order that they occurred. A black dotted line represents a transfer found with majority rule, a red dotted line is a transfer found using PrimConsTree version 2 and a solid black line stand for both.}
\label{fig:results}
\end{figure*}

Transfers found with extended majority rule are black dotted lines, the ones found with PrimConsTree version 2 are red dotted lines, and a solid line stands for both. To refer to a transfer, we prefix its number with $Px-$ if it was inferred with PrimConsTree version 2 and $Mx-$ if it was inferred with the extended majority rule, $x$ being the number of the cluster (1 to 5). Three of the transfers found with PrimConsTree version 2 were already validated by Boc et al. in \cite{boc2013inferring}. 
In particular, the equivalent transfers P4-2 and M5-1, predicted in Boc et al. 2010 (see Fig. 6 in \cite{boc2010inferring}) were found by both PrimConsTree version 2 and the extended majority rule. The majority rule additionally revealed M3-2 (or its equivalent M5-4) that were predicted in the same analysis. 
The transfers P1-1 (or its equivalent P3-2 and P5-2), P2-2, and P4-2 has been predicted by Boc et al. 2013 (see Fig. 3 in \cite{boc2013inferring}). The latter were also found by Tahiri et al. 2018, using the extended majority rule, respectively as M4-2, M2-2 and M5-1 in Figure \ref{fig:results}. Only the majority rule was able to find M1-1 (or its equivalent M2-2 and M4-1), M3-2 (or its equivalent M5-4) and M4-3.
Finally, Tahiri et al. 2018 found M3-1 and M3-2 (or its equivalent M5-4) using the extended majority rule, those were predicted by Boc et al. 2013 (see Fig. 2b in \cite{boc2013inferring}). 
Moreover, we identified a higher number of gene transfers, indicating that a non-parsimonious model, involving both a split and a fusion event, provides the most likely explanation for the evolutionary history of Archaeabacteria organisms, consistent with the hypothesis proposed by \cite{matte2002archaeal}.

\section{Conclusions}

In this paper, we extended PrimConsTree version 1, a graph-based approach for constructing consensus trees with balanced branch lengths. To achieve this, we clustered the input trees and proposed new key criteria, such as clade frequency, to derive the maximum spanning tree. We provide a detailed description of the supergraph construction and an enhanced version of the well-established Prim algorithm.

We utilized the PrimConsTree version 2 method to infer horizontal gene transfer among a group of Archaeabacteria organisms. A comparison of our results with those obtained using the extended majority rule identified interesting transfers but did not indicate a superiority of our method.

For future research directions, we propose exploring the extension of the applicability of PrimConsTree version 2 to address the supertree problem, particularly in situations where the number of leaves in gene trees may vary. Furthermore, there is potential for additional refinements to enhance the ability of the algorithm to generate a consensus tree that closely aligns with the gene trees.

\section*{Declarations}

\subsection*{Ethics approval and consent to participate}
Not applicable.

\subsection*{Consent for publication}
Not applicable.

\subsection*{Availability of data and material}
The source code for PrimConsTree (Python v 3.11), along with all datasets utilized in this study, are openly accessible and publicly available on GitHub at \url{https://github.com/tahiri-lab/PrimConsTree}.

\subsection*{Competing interests}
No competing interest is declared.

\subsection*{Funding} 
This work is supported in part This research was funded by the Natural Sciences, Engineering Research Council of Canada - Discovery Grants RGPIN-2022-04322 and Fonds de recherche du Québec - Nature and Technologies 326911, and JSPS KAKENHI grant 22H03550/23K24807.

\subsection*{Author contribution}
All authors contributed equally to the development of this article.

\bibliographystyle{unsrt}  
\bibliography{references}

\begin{thebibliography}{10}

\bibitem{torquet2005icbbb}
Elio Torquet, Jesper Jansson, and Nadia Tahiri.
\newblock Graph-based method for constructing consensus trees.
\newblock In {\em BIO Web of Conferences}, volume 163, page 01004. EDP Sciences, 2025.

\bibitem{o2018efficacy}
Joseph~E O’Reilly and Philip~CJ Donoghue.
\newblock The efficacy of consensus tree methods for summarizing phylogenetic relationships from a posterior sample of trees estimated from morphological data.
\newblock {\em Systematic biology}, 67(2):354--362, 2018.

\bibitem{degnan2009properties}
James~H Degnan, Michael DeGiorgio, David Bryant, and Noah~A Rosenberg.
\newblock Properties of consensus methods for inferring species trees from gene trees.
\newblock {\em Systematic Biology}, 58(1):35--54, 2009.

\bibitem{tahiri2022building}
Nadia Tahiri, Bernard Fichet, and Vladimir Makarenkov.
\newblock Building alternative consensus trees and supertrees using k-means and robinson and foulds distance.
\newblock {\em Bioinformatics}, 38(13):3367--3376, 2022.

\bibitem{aguse2019summarizing}
Nuraini Aguse, Yuanyuan Qi, and Mohammed El-Kebir.
\newblock Summarizing the solution space in tumor phylogeny inference by multiple consensus trees.
\newblock {\em Bioinformatics}, 35(14):i408--i416, 2019.

\bibitem{tahiri2018new}
Nadia Tahiri, Matthieu Willems, and Vladimir Makarenkov.
\newblock A new fast method for inferring multiple consensus trees using k-medoids.
\newblock {\em BMC evolutionary biology}, 18:1--12, 2018.

\bibitem{felsenstein1985phylogenies}
Joseph Felsenstein.
\newblock Phylogenies and the comparative method.
\newblock {\em The American Naturalist}, 125(1):1--15, 1985.

\bibitem{hahn2005estimating}
Matthew~W Hahn, Tijl De~Bie, Jason~E Stajich, Chi Nguyen, and Nello Cristianini.
\newblock Estimating the tempo and mode of gene family evolution from comparative genomic data.
\newblock {\em Genome research}, 15(8):1153--1160, 2005.

\bibitem{kosakovsky2005not}
Sergei~L Kosakovsky~Pond and Simon~DW Frost.
\newblock Not so different after all: a comparison of methods for detecting amino acid sites under selection.
\newblock {\em Molecular biology and evolution}, 22(5):1208--1222, 2005.

\bibitem{volz2013viral}
Erik~M Volz, Katia Koelle, and Trevor Bedford.
\newblock Viral phylodynamics.
\newblock {\em PLoS computational biology}, 9(3):e1002947, 2013.

\bibitem{lefort2015fastme}
Vincent Lefort, Richard Desper, and Olivier Gascuel.
\newblock Fastme 2.0: a comprehensive, accurate, and fast distance-based phylogeny inference program.
\newblock {\em Molecular biology and evolution}, 32(10):2798--2800, 2015.

\bibitem{rannala2015art}
Bruce Rannala.
\newblock The art and science of species delimitation.
\newblock {\em Current Zoology}, 61(5):846--853, 2015.

\bibitem{sifat2024new}
Md~Habibur~Rahman Sifat and Nadia Tahiri.
\newblock A new algorithm for building comprehensive consensus tree.
\newblock In {\em Graphs and more Complex structures for Learning and Reasoning: Proceedings of the 38th Annual AAAI Conference on Artificial Intelligence}, 2024.

\bibitem{bryant2003classification}
David Bryant.
\newblock A classification of consensus methods for phylogenetics.
\newblock {\em DIMACS series in discrete mathematics and theoretical computer science}, 61:163--184, 2003.

\bibitem{jansson2016improved}
Jesper Jansson, Chuanqi Shen, and Wing-Kin Sung.
\newblock Improved algorithms for constructing consensus trees.
\newblock {\em Journal of the ACM (JACM)}, 63(3):1--24, 2016.

\bibitem{sokal1981taxonomic}
Robert~R Sokal and F~James Rohlf.
\newblock Taxonomic congruence in the {Leptopodomorpha} re-examined.
\newblock {\em Systematic Zoology}, 30(3):309--325, 1981.

\bibitem{wilkinson2001efficiency}
Mark Wilkinson and Joseph~L Thorley.
\newblock Efficiency of strict consensus trees.
\newblock {\em Systematic Biology}, 50(4):610--613, 2001.

\bibitem{wilkinson1996majority}
Mark Wilkinson.
\newblock Majority-rule reduced consensus trees and their use in bootstrapping.
\newblock {\em Molecular Biology and evolution}, 13(3):437--444, 1996.

\bibitem{margush1981consensusn}
Timothy Margush and Fred~R McMorris.
\newblock Consensus {$n$}-trees.
\newblock {\em Bulletin of Mathematical Biology}, 43(2):239--244, 1981.

\bibitem{goloboff2003improvements}
Pablo~A Goloboff, James~S Farris, Mari K{\"a}llersj{\"o}, Bengt Oxelman, Mart{\'\i}n~J Ram{\'\i}rez, and Claudia~A Szumik.
\newblock Improvements to resampling measures of group support.
\newblock {\em Cladistics}, 19(4):324--332, 2003.

\bibitem{G05}
Pablo~A. Goloboff.
\newblock Minority rule supertrees? {MRP}, {Compatibility}, and {Minimum Flip} may display the {\emph{least}} frequent groups.
\newblock {\em Cladistics}, 21(3):282--294, 2005.

\bibitem{STACS2024}
Jesper Jansson, Wing-Kin Sung, Seyed~Ali Tabatabaee, and Yutong Yang.
\newblock A faster algorithm for constructing the frequency difference consensus tree.
\newblock In {\em 41st International Symposium on Theoretical Aspects of Computer Science (STACS 2024)}, pages 43--1. Schloss Dagstuhl--Leibniz-Zentrum f{\"u}r Informatik, 2024.

\bibitem{V_19}
Joel~D. Velasco.
\newblock The foundations of concordance views of phylogeny.
\newblock {\em Philosophy, Theory, and Practice in Biology}, 11(20), 2019.

\bibitem{makarenkov2023inferring}
Vladimir Makarenkov, Gayane~S Barseghyan, and Nadia Tahiri.
\newblock Inferring multiple consensus trees and supertrees using clustering: A review.
\newblock {\em Data Analysis and Optimization: In Honor of Boris Mirkin's 80th Birthday}, pages 191--213, 2023.

\bibitem{jarnik1930jistem}
Vojt{\v{e}}ch Jarn{\'\i}k.
\newblock O jist{\'e}m probl{\'e}mu minim{\'a}ln{\'\i}m.
\newblock {\em Pr{\'a}ca Moravsk{\'e} Pr{\'\i}rodovedeck{\'e} Spolecnosti}, 6:57--63, 1930.

\bibitem{rousseeuw1987silhouettes}
Peter~J Rousseeuw.
\newblock Silhouettes: a graphical aid to the interpretation and validation of cluster analysis.
\newblock {\em Journal of computational and applied mathematics}, 20:53--65, 1987.

\bibitem{woodhams2016simulating}
Michael~D Woodhams, Peter~J Lockhart, and Barbara~R Holland.
\newblock Simulating and summarizing sources of gene tree incongruence.
\newblock {\em Genome biology and evolution}, 8(5):1299--1315, 2016.

\bibitem{kendall2016mapping}
Michelle Kendall and Caroline Colijn.
\newblock Mapping phylogenetic trees to reveal distinct patterns of evolution.
\newblock {\em Molecular biology and evolution}, 33(10):2735--2743, 2016.

\bibitem{cock2009biopython}
Peter~JA Cock, Tiago Antao, Jeffrey~T Chang, Brad~A Chapman, Cymon~J Cox, Andrew Dalke, Iddo Friedberg, Thomas Hamelryck, Frank Kauff, Bartek Wilczynski, et~al.
\newblock Biopython: freely available python tools for computational molecular biology and bioinformatics.
\newblock {\em Bioinformatics}, 25(11):1422--1423, 2009.

\bibitem{matte2002archaeal}
Oriane Matte-Tailliez, C{\'e}line Brochier, Patrick Forterre, and Herv{\'e} Philippe.
\newblock Archaeal phylogeny based on ribosomal proteins.
\newblock {\em Molecular biology and evolution}, 19(5):631--639, 2002.

\bibitem{boc2010inferring}
Alix Boc, Herv{\'e} Philippe, and Vladimir Makarenkov.
\newblock Inferring and validating horizontal gene transfer events using bipartition dissimilarity.
\newblock {\em Systematic biology}, 59(2):195--211, 2010.

\bibitem{boc2012t}
Alix Boc, Alpha~Boubacar Diallo, and Vladimir Makarenkov.
\newblock T-rex: a web server for inferring, validating and visualizing phylogenetic trees and networks.
\newblock {\em Nucleic acids research}, 40(W1):W573--W579, 2012.

\bibitem{Wiens2008}
John~J Wiens, Caitlin~A Kuczynski, Sarah~A Smith, Daniel~G Mulcahy, Jr. Sites, Jack~W, Ted~M Townsend, and Tod~W Reeder.
\newblock {Branch Lengths, Support, and Congruence: Testing the Phylogenomic Approach with 20 Nuclear Loci in Snakes}.
\newblock {\em Systematic Biology}, 57(3):420--431, 2008.

\bibitem{boc2013inferring}
Alix Boc, Pierre Legendre, and Vladimir Makarenkov.
\newblock An efficient algorithm for the detection and classification of horizontal gene transfer events and identification of mosaic genes.
\newblock {\em Studies in Classification, Data Analysis, and Knowledge Organization}, pages 253--260, 2013.

\end{thebibliography}

\newpage
\appendix

\section{Algorithms}
Algorithm \ref{alg:primleaves} depicts how the vertices that must be leaf nodes in the MST are connected.

\begin{algorithm}[!]
\caption{\textsc{AttachLeaves}: Incorporate vertices supposed to be leaf nodes in a maximum spanning tree}
\label{alg:primleaves}
\begin{algorithmic}
\State \textbf{Input:} $\mathcal{M}$ the maximum spanning tree of $G[I(\mathcal{G})]$, \\
\hspace{0.9cm} $\mathcal{G}$ the super-graph with $V(\mathcal{G}) = I(\mathcal{G}) \cup L(\mathcal{G})$, \\
\hspace{0.9cm} edge frequencies $W$, and \\
\hspace{0.9cm} clade frequencies $F$
\State \textbf{Output:} $\mathcal{M}$ with nodes $L(\mathcal{G})$ attached
\end{algorithmic}
\begin{algorithmic}[1]
\Function{AttachLeaves}{$\mathcal{M}$, $\mathcal{G}$}
    \ForAll{node $l \in  L(\mathcal{G})$}
        \State $e^* \gets \textbf{Null}$
        \State $C \gets$ \{$(v, l) \in E(\mathcal{G})$\}
        \ForAll{edge $(l, v) \in C$}
            \If{$W(v, l) > W(e^*)$}
                \State $e^* \gets (v, l)$
            \ElsIf{$W(v, l) = W(e^*)$}
                \If{$F(v) > F(e^*[0])$}
                    \State $e^* \gets (v, l)$
                \EndIf
            \EndIf
        \EndFor
        \State $V(\mathcal{M}) \gets V(\mathcal{M}) \cup \{l\}$
        \State $E(\mathcal{M}) \gets E(\mathcal{M}) \cup \{e^*\}$
    \EndFor
    \State \Return $\mathcal{M}$
\EndFunction
\end{algorithmic}
\end{algorithm}

\end{document}